\documentclass[12pt,a4paper]{article}
\usepackage{amsmath}
\usepackage{amsfonts}
\usepackage{amssymb}
\usepackage{amsthm}
\newtheorem{propi}{Proposition}

\newtheorem{teo}{Theorem}

\newtheorem{coro}{Corollary}

\newtheorem{defi}{Definition}
\newtheorem{exa}{Example}
\begin{document}
\title{Alternative Axioms in Group Identification Problems\thanks{Carlos Pimienta (UNSW, Australia) and Alejandro Neme (UNSL, Argentina) provided useful comments and criticisms on a previous version of this paper.}}

\author{Federico Fioravanti \ \ \ \ Fernando Tohm\'e \\
INMABB, Universidad Nacional del Sur\\
CONICET \\
Av. Alem 1253, (8000) Bah\'{\i}a Blanca, Argentina\\
e-mail:{\tt federico.fioravanti@uns.edu.ar},  {\tt ftohme@criba.edu.ar}, \\
phone: 54-291-4359890 \ \ \ 54-291-4595138\\}

\date{}
\maketitle

\begin{abstract}
Kasher and Rubinstein (1997) introduced the problem of classifying the members of a group in terms of the opinions of their potential members. This involves a finite set of agents $N=\{1,2,\ldots, n\}$, each one having an opinion about which agents should be classified as belonging to a specific subgroup $J$. A {\em Collective Identity Function} (CIF) aggregates those opinions yielding the class of members deemed $J$. Kasher and Rubinstein postulate axioms, intended to ensure fair and socially desirable outcomes, characterizing different CIFs.
We follow their lead by replacing their liberal axiom by other axioms, constraining the spheres of influence of the agents. We show that some of them lead to different CIFs while in another instance we find an impossibility result.\\
\textbf{Keywords:} Group Identification; Social Choice; Decisiveness; Classification; Liberalism.
\end{abstract}
\section{Introduction}
People, countries and  inanimate objects, among other entities, are customarily classified in groups.
Sometimes these classifications are simple and obvious, as for instance organizing countries by the continent to which they belong. 
However, if we want to identify the members of a particular community or, say, regroup countries in terms of their degree of ``eco-friendliness'', the classification is far from evident, and the assessment of the individuals or nations involved matters for the final result. 
Kasher and Rubinstein (K-R) analyzed this problem (``Who is a J?'', 1997), presenting it as a question of defining appropriate aggregation functions over profiles of opinions. 
Each person in a society is assumed to have an opinion about which individuals, including theirself, may or not belong to the group. The opinion of an individual is thus identified with a subset of the class of all agents. Then, to determine the identities of the individuals who will finally deemed to belong to the group, the opinions of all the individuals are aggregated. There exist many ways in which this aggregation can be carried out, each embodied in an aggregator function, which they call a Collective Identity Function (CIF).
K-R characterize three different CIFs, each of which reflects a particular notion of ``fairness''. 
The ``Liberal'' CIF classifies as a $J$ any individual that sees theirself as being a member of $J$; the ``Dictatorial'' CIF is such that a single individual decides who is in $J$, and the ``Oligarchic'' one, in which this decision is made by a specific group of individuals. 
K-R's original contribution started a line of work extending and modifying it. So, for instance, Saporiti (2012), modified K-R's axioms; Sung and Dimitrov (2003, who provided new characterizations; Cho and Ju (2016) that work on the identification of two groups; Cho and Saporiti (2015), dealing with the incentive aspects of declaring opinions on who belongs to $J$, and Fioravanti and Tohm\'e (2019) using an infinite setting of agents.\\

The axioms that K-R used are classical ones: {\it Monotonicity} (if an individual $i$ belongs already to $J$ under some profile of opinions, if in another profile there are more agents favoring their inclusion in $J$, $i$ has to belong to $J$ in this new profile), {\it Consensus} (if everybody thinks $i$ is in $J$ then $i$ must be a $J$), {\it Symmetry} (if the number of agents that think that $i$ belongs to $J$ is the same as that of those to think that $j \in J$, then either both or none are in $J$) and {\it Independence} (only the opinions on whether $i$ belongs or not to $J$ matter for determining if $i$ is a $J$ or not). These axioms are expected to be satisfied by any ``fair'' aggregator.\\ 

K-R added to them a new axiom, the Liberal Principle, capturing a specific notion of ``liberalism'', stating that if someone considers theirself a member of $J$, then the set of $J$ cannot be empty, while if someone considers theirself not a member of $J$, then not everybody can be classified as a $J$. In this paper we explore the consequences of replacing this axiom by four alternatives axioms.  
In particular, three of our four axioms can be seen as opposed to liberalism. They give the power to  agents to decide whether {\em another} agent should or not belong to $J$. The fourth axiom, instead, is a stronger version of K-R's Liberal Principle.\\

First, we present an axiom that grants each agent the right to be {\em decisive} on the question of whether another given individual belongs or not to $J$. A weaker version of this axiom amounts to allowing at least two agents to decide whether or not another particular agent belongs to the group. We prove that the Strong Liberal CIF is the only CIF that verifies these two axioms.
The other axiom states that at least two agents can be {\em semidecisive} over two other agents (i.e. the positive opinion on whether these agents belong to $J$ determines the actual social decision). In this case we find that there are more CIFs verifying it.\\

These three axioms abstract away rules that are considered natural in societies. For example, it is assumed that parents have the right to decide which religious or ethnical allegiance they choose for their children. No external power can force them to give their children a particular identity.\\

Finally, we consider an extreme version of the Liberal Principle. It has a positive part prescribing that if there is an agent that thinks that some (possibly another) agent is in $J$, then someone must be a $J$. The negative part prescribes that if there is an agent that thinks that there is another agent who is not in $J$, then not everyone can be a $J$. This axiom can be interpreted as indicating that at least some opinions actually matter. The failure of this axiom could mean that some results are obtained against the opinion of all the agents, a clear illiberal situation. We prove that there is no CIF verifying the full version of this axiom, but we find uniqueness results for two CIFs, one verifying its positive and the other its negative part.\\

This paper should not be interpreted in a normative sense. Our main goal is to explore alternative axiomatizations and analyze the properties of the ensuing aggregation functions (if they exist). Despite this purely technical goal, we think that some of the axioms we introduce here can be seen as highly stylized versions of spheres of individual decision-making permitted in human societies. A word of caution is needed here: this paper is a pure exercise of {\em characterization of an axiomatic system}. In this sense it is not intended for any practical (i.e. {\em normative}) use. Its gist is both in the characterization of CIFs and in the proofs of impossibility of satisfaction of a set of axioms.\\

We describe the original axioms in the next Section, while in Section 3 we focus on the Strong Liberal CIF and obtain a characterization of the system with the new set of axioms.
The results on the relation with the extreme version of liberalism are presented in section 4. Section 5 concludes.

\section{Basic notation and axioms}
We consider a set $N$ of individuals, with $\mid N\mid = n <\infty$. 
Each individual $i$ has an ``opinion'' described by a set $C_{i}\subseteq N$. By a slight abuse of language we will use $J$ to denote a subset of $N$ and a {\em type} of individual. With $j\in C_{i}$ we indicate that $i$ thinks that $j$ has to be classified as a $J$. 
On the contrary, if $j\notin C_{i}$, then $i$ does not think that $j$ should be a $J$.
A profile of opinions is a tuple $C=(C_{1},\ldots, C_{n})$ where $C_{i}\subseteq N$ for every $i\in N$. 
Let $P(N)$ be the power set of $N$ and $\textbf{C}$ the set of all possible profiles of opinions, i.e., $\textbf{C}= P(N)^{n}$. 
By a further abuse of language we denote by $J$ the Collective Identity Function (CIF) $J:\textbf{C}\rightarrow P(N)$. 
That is, a CIF is a function that assigns to each profile of opinions a set $J \subseteq N$, of the individuals that will be socially considered to be $J$. \\
Let us now present the axioms introduced by K-R to capture properties that a fair CIF should satisfy:\\
The first one states that if a CIF classifies $i$ as a $J$, then having more people considering $i$ as a $J$ will just reinforce $i$ status. 
The same is true if the CIF classifies an agent $j$ as not being in $J$: if $j$ ``loses'' votes, they keep being a non-$J$.
\begin{itemize}
\item \textbf{Monotonicity}(MON): let $i\in J(C)$ and $C'$ be a profile identical to $C$ except that there are individuals $i$ and $k$, such that $i\notin C_{k}$ and $i\in C'_{k}$; then $i\in J(C')$. 
Analogously, if $i\notin J(C)$ and $ C'$ is identical to $C$, except that there is a $k$ such that if $i\in C_{k}$ and $i\notin C'_{k}$, then $i\notin J(C')$.
\end{itemize}
The next axiom amounts to claim that if everybody thinks that $i$ is a $J$, then the CIF has to classify them as a $J$. 
On the other hand, if no one considers that $i$ is a $J$, then the CIF cannot allow them to be a $J$.
\begin{itemize}\item \textbf{Consensus}(C): if $j\in C_{i}$ for all $i$, then $j\in J(C)$; if $j\notin C_{i}$ for all $i$, then $j\notin J(C)$.
\end{itemize}
The axiom of Symmetry indicates that the aggregator does not classify any two members of the society on any basis other than that embedded in the profile of views. That is, if two agents are ``similar'', the CIF must classify them in the same way.
\begin{itemize}\item \textbf{Symmetry}(SYM): $j$ and $k$ are symmetric in a profile $C$ if
\begin {itemize}
\item [(i)] $C_{j}-\{j,k\}=C_{k}-\{j,k\}$
\item [(ii)] for all $i\in N-\{j,k\}, j\in C_{i}$ iff $k\in C_{i}$
\item [(iii)] $j\in C_{j}$ iff $k\in C_{k}$
\item [(iv)] $j\in C_{k}$ iff $k\in C_{j}$
\end {itemize}
Then, $j\in J(C)$ if and only if $k\in J(C)$.
\end{itemize}
The final axiom ensuring a fair aggregation process indicates that to define whether $i$ is or not a $J$, a CIF must take into account only opinions about $i$.
\begin{itemize}\item \textbf{Independence}(I): consider two profiles $C$ and $ C'$ and let $i$ be an individual in $N$. 
If for all $k\in N$, $i\in C_{k}$ if and only if $i\in C'_{k}$, then $i\in J(C)$ if and only if $i\in J(C')$.\footnote{In their work, K-R use this version to characterize the Dictatorship and Oligarchic aggregator, and use a weaker version to characterize the Liberal aggregator. The latter version ensures their results. 
We use only this definition because it is a standard one in the Social Choice literature (Arrow, 1963; Rubinstein and Fishburn, 1986; Nicolas, 2007; Cho and Ju, 2017).}
\end{itemize}
Kasher and Rubinstein introduce another axiom to capture the idea that an individual's view of themself should be considered relevant. 
\begin{itemize}\item \textbf{The Liberal Principle}(L): if there exists an $i\in N$ such that $i\in C_{i}$, then $J(C)\neq \emptyset$; and if there exists an $i\in N$ such that $i\notin C_{i}$, then $J(C)\neq N$.
\end{itemize}
That is, if someone thinks they are in $J$, then {\em somebody} must be a $J$, and if someone thinks they are not in $J$, then {\em not everyone} can be a $J$.\\ 
There are some cases where our axioms may be more appropriate. For example, in a classroom, where the opinion that a student has about themself should not determine his standing.
Of course, what one thinks about oneself is relevant, but in a variety of cases others may be better judges of individual traits. Moreover, the individuals that have the right to judge should have also the right to determine the outcome. For example, suppose the situation where a family is in an amusement park and wants to ride a rollercoaster. It seems reasonable to assume that, the opinions of the children notwithstanding, whether the little baby goes on the ride or not it is the parents' decision to make. Our axioms intend to capture this intuition.

We introduce the following definitions, which will be useful for the rest of this work:
\begin{defi}
We say that $i\in N$ is {\em decisive} over $j\in N$, if the following is true: $j\in C_{i}$ iff $j\in J(C)$. 
On the other hand $i\in N$ is {\em semidecisive} over $j\in N$, if only one of the following is true: if $j\in C_{i}$ then $j\in J(C)$ or, if $j\notin C_{i}$, then $j\notin J(C)$.
\end{defi}

It is clear that if an individual is decisive then is also semidecisive, while the converse is not true.\\

The first axiom  we introduce allows every agent to be decisive over some other agent (possibly themself).
The second and third axioms, are weaker versions of the first one, allowing at least two agents to be decisive or semidecisive over some other (although they could be themselves) two agents.\\

Now we present the formal definitions of these three axioms, which abstract away and generalize the aforementioned intuitions:
\begin{itemize}
\item \textbf{Decisiveness}(D): for each $i\in N$ there exists a $j\in N$ such that $i$ is decisive over $j$.
\item \textbf{Minimal Decisiveness}(MD): there exist at least two distinct individuals, $i$ and $j$, and two distinct individuals, $k$ and $l$, such that $i$ is decisive over $k$ and $j$ is decisive over $l$.
\item \textbf{Minimal Semi-Decisiveness}(MSD):  there exist at least two distinct individuals, $i$ and $j$, and two distinct individuals, $k$ and $l$, such that $i$ is semidecisive over $k$ and $j$ is semidecisive over $l$.
\end{itemize}
The remaining axiom is an extreme version of the Liberal Principle used by K-R. 
It states the impossibility of $J = \emptyset$, if somebody thinks that there exists at least one individual who should be classified as a $J$.
Similarly, if there is some individual that thinks that there exists someone who is not a $J$, then not everybody can be considered to be a $J$. 
That is, a single positive opinion is enough to indicate that in fact some $J$ exists and a single negative opinion suffices to ensure that there exist at least one non-$J$. These two principles seem very natural in liberal societies in which some people can make decisions on behalf of others (parents for their children, teachers for their students, etc.). \\
\begin{itemize}
\item \textbf{Extreme Liberalism}(EL):
\begin{itemize}
\item[(i)] There exists a pair $\{i,j\}\subseteq N$ such that if $j\in C_{i}$, then $J\neq \emptyset$.
\item[(ii)] There exists a pair $\{i,j\}\subseteq N$ such that if $j\notin C_{i}$, then $J\neq N$.
\end{itemize}
\end{itemize}

\section{Strong Liberal CIF}
Kasher and Rubinstein introduce in their work the following CIF:\begin{center}
$\mathit{StL}(C)=\{i: i\in C_{i}\}$
\end{center}
This aggregator classifies as a $J$ every agent that considers theirself a $J$.
They characterize this CIF as the only one that satisfies (MON), (C), (SYM), (I) and (L). 
Sung and Dimitrov (2003) refine this result, and give a characterization of the Strong Liberal CIF as the only CIF that verifies (SYM), (I) and (L).
Because of its uniqueness, they conclude that if a CIF verifies (SYM), (I) and (L) then it also verifies (MON) and (C). 
In this section, we change the axiom (L) for (D), (MD) or (MSD), and explore which CIFs verify them. 
We think of (D) as a new axiom capturing an opposite notion of liberalism than (L). \\
The next example shows that if a CIF satisfies (L) it does not necessarily verify (D).
\begin{exa}
Let $J(C)=S\in P(N)$ for all $C\in\textbf{C}$ with $S\notin\{\emptyset, N\}$. 
This CIF verifies (L) because its outcome is neither the empty set nor the entire group. And it does not verify (D). 
Suppose $1$ is decisive over $2$ and $S=\{1,2\}$. 
If $C=(\{1\},\{2\})$, then $J(C)=\{1,2\}$, contradicting that $1$ is decisive over $2$.   
\end{exa}
Another example shows that (D) does not imply (L).
\begin{exa}
Suppose $N=\{1,2\}$ and $J$ is a CIF such that $1$ is decisive over $2$ and $2$ decisive over $1$.
It verifies (D) because every agent is decisive over the other agent.
It does not verify (L) because $J(\{1\},\{2\})=\emptyset$ even if $i\in C_{i}$ for $i=1,2$.
\end{exa} 
Although that it seems that (L) and (D) capture opposite notions, the uniqueness of the Strong Liberal CIF is still valid when we replace one for the other.
Even more, we can drop one of the axioms to show this:
\begin{teo}
For $N$ such that $|N|>2$, the only CIF that verifies (SYM) and (D) is the Strong Liberal CIF.\footnote{When $|N|=2$, we find in Example 2 a CIF that verifies (SYM) and (D).}
\end{teo}
\begin{proof}
First, we show that two agents cannot be decisive over the same agent. 
Suppose, on the contrary, that $i$ and $j$ are decisive over $k$. 
If in the profile $C$, we have $k\in C_{i} $ and $k\notin C_{j} $, then $i\in J(C)$ and $i\notin J(C)$, a contradiction. 
Since we assume that the group of voters is finite, for every CIF that satisfies (D), any agent can be decisive over just one agent. 
Let $ \sigma $ be a permutation of the set $N$ of individuals.
If a CIF verifies that every agent is decisive over any other agent (including themself), it must have the following form:
$$ J_{\sigma}(C)=\{\sigma (i): \sigma (i)\in C_{i}\} $$
If $\sigma$ is the identity permutation we obtain the Strong Liberal CIF, satisfying both axioms.\\
Now consider a permutation $\sigma$ such that $\sigma(i)\neq i$ for some $ i $. 
Suppose that $J_{\sigma}$ satisfies (SYM). 
Then, for every possible profile $(C_{1}, \ldots, C_{N})$, if there exists a pair $j,k$ of symmetric individuals, either $\{j,k\}\subseteq J_{\sigma}(C)$ or $\{j,k\}\nsubseteq J_{\sigma}(C)$. Consider a simple example in which $N=\{1,2,3\}$ and $\sigma =(23)$ (the permutation that exchanges $2$ and $3$, leaving $1$ fixed). 
Consider the profile $C=(\{3\},\{1,3\},\{1\})$. 
Agents $1$ and $3$ are symmetric. 
By definition $J_{\sigma}(C)=\{3\}$, but according to (SYM) it should be that $1\in J_{\sigma}(C)$ if and only if $3\in J_{\sigma}(C)$. Contradiction.
This shows that the CIF verifies (D) but not (SYM).
\end{proof}
Since the $\mathit{StL}$ CIF verifies (MON), (C) and (I), the following corollary can be obtained immediately from Theorem 1.
\begin{coro}
If a CIF satisfies (SYM) and (D), it also verifies (MON), (C) and (I).
\end{coro} 
Furthermore, we can get a refinement of this result.
\begin{propi}
If a CIF verifies (D), it also verifies (MON), (C) and (I).
\end{propi}
\begin{proof}
Suppose $j$ is decisive over $i$.
\begin{itemize}
\item (D) implies (MON). 
If we have $i\in J(C)$, then it must be the case that $i\in C_{j}$. 
Consider a profile $C'$ identical to $C$ except that there is an individual $k$ such that $i\notin C_{k}$ and $i\in C'_{k}$. 
There is no change in the opinion of $j$, so by (D) we have that $i\in J(C')$.
The same happens if $i\notin J(C)$.
\item (D) implies (C). 
If $i\in C_{k}$ for all $k\in N$, in particular it is true that $i\in C_{j}$. 
So by (D), we have $i\in J(C)$. 
If $i\notin C_{k}$ for all $k\in N$, we have that, in particular $i\notin C_{j}$. 
Again, by (D) we have $i\notin J(C)$.
\item (D) implies (I).
Let $C,C'\in\textbf{C}$ be two profiles such that for all $k\in N$, $i\in C_{k}$ if and only if $i\in C'_{k}$. 
If we have that $i\in J(C)$, then it must be the case that $i\in C_{j}$, so $ i\in C'_{j}$. 
Then by (SL) we have $i\in J(C')$. 
The same happens if $i\notin J(C)$.
\end{itemize}
\end{proof}
An example illustrates that (SYM) does not imply any of the other axioms:\footnote{This example is used in Sung and Dimitrov (2003) to introduce a CIF that verifies (SYM) and (L) but not (I).}
\begin{exa}
Define the following CIF:
$$J(C)= \left\{ \begin{array}{lcc}
             \mathit{StL}(C) &   if  & \mathit{StL(C)}\in\{\emptyset, N\} \\
             \\ N-\mathit{StL}(C) & & otherwise \\            
             \end{array}
\right.   $$
This CIF verifies (SYM), but not (MON), (C), (I) and (D).
\end{exa}
Minimal Decisiveness is a weaker version of Decisiveness, since it only requires that at least two agents must be decisive. 
In the proof of \mbox{Theorem 1}, we do not make any reference to the number of agents who are decisive.
So weakening the axiom does not change the result. 
We obtain the following:
\begin{coro}
The only CIF that verifies (SYM) and (MD) is the Strong Liberal CIF.
\end{coro}
The results from \mbox{Corollary 1} and \mbox{Proposition 1} remain the same if we replace (D) by (MD).\\ 
We give two examples of CIF's that will be useful in the future.\footnote{Samet and Schmeidler (2003) defines this as consent rules. In particular, the $U$ CIF is the consent rule where $s=n$ and $t=1$, and the $\mathit{In}$ CIF is the consent rule where $s=1$ and $t=n$.}
The \textit{Unanimity} $(\mathit{U})$ CIF, that prescribes that $i$ is a $J$ if and only if everybody thinks that $i$ is a $J$. 
Formally:
$$U(C)=\{i: i\in C_{k} \ for\ all\ k\in N\}$$
And the \textit{Inclusive} $(\mathit{In})$ CIF, that indicates that $i$ is a $J$ if someone thinks that $i$ is a $J$:
$$\mathit{In}(C)=\{i: i\in C_{k} \ for\ some\ k\in N\}$$ 
Minimal Semi-Decisiveness is weaker than Minimal Decisiveness, since it only requires that at least two agents must be semidecisive.
When we use such a weak notion, we find more CIFs that satisfy (SYM) and (MSD). Moreover, they also verify (MON), (C) and (I):
\begin{propi}
The Strong Liberal CIF is not the only CIF that verifies (MON), (C), (SYM), (I) and (MSD).
\end{propi}
\begin{proof}
It is clear that the Strong Liberal CIF verifies the proposed axioms. 
The $U$ and $\mathit{In}$ are two examples.
Both CIFs satisfy all the axioms. With $U$, every agent is semidecisive since, if $i\notin C_{j}$ for some $j\in N$, then $i\notin U(C)$. 
With $\mathit{In}$, every agent is semidecisive since, if $i\in C_{j}$ for some $j\in N$, then $i\in In(C)$.
\end{proof}
\section{Extreme Liberalism}
Let us recall the definition of this axiom:\\\\\textbf{Extreme Liberalism}(EL):
\begin{itemize}
\item[(i)] If there exists a pair $\{i,j\}\subseteq N$ such that $j\in C_{i}$, then $J\neq \emptyset$.
\item[(ii)] If there exists a pair $\{i,j\}\subseteq N$ such that $j\notin C_{i}$, then $J\neq N$.
\end{itemize} 
From the definitions, we can see that (EL) implies (L). 
The idea that (EL) captures, is that if someone thinks that there is at least one $J$, then the group of socially accepted people cannot not be empty.
Similarly, if somebody thinks that some individual is not a $J$, then not everybody can be a $J$.
The Liberal Principle (L) captures the idea that the opinion about oneself is important.
(EL) goes further, stating that the opinion of any individual about somebody else matters.
When we turn to such an extreme concept of ``liberalism'', we obtain an impossibility result.
This is easy to see from the fact that (EL) implies (L) and the Strong Liberal CIF does not verify (EL).
Beside this issue, we can show the uniqueness of CIFs that satisfy parts (i) and (ii) of (EL) separately.
 \begin{teo}
The Inclusive CIF $\mathit{In}$ is the only CIF that verifies (MON), (C), (I) and (EL)(i).\footnote{The Inclusive CIF $\mathit{In}$ is defined in Section 3.}
\end{teo}
\begin{proof}
First of all, notice that $\mathit{In}$ satisfies these four axioms. 
Consider a different CIF $J$ that also verifies them. 
Suppose that there is a profile $C$ such that $i\in C_{k} $ for some $k$ but $i\notin J(C)$.
By applying (MON) several times, we can find a profile $C'$ identical to $C$ with the exception that for every $j\neq k$, $i\notin C'_{j}$ so that $i\notin J(C')$. 
Let $C''$ be the profile such that $C''_{k}=\{i\}$ and $C''_{j}=\emptyset$ for all $j\neq k$. 
By (C), $J(C'')$ cannot include any individual from $N-\{i\}$. 
So $J(C'')\in\{\emptyset, \{i\}\}$. 
Because this CIF satisfies part (i) of (EL), it follows that $J(C'') \neq \emptyset$ and thus, $J(C'')=\{i\}$. 
But we obtain a contradiction with (I), because agent $i$ is treated in the same way in profiles $C'$ and $C''$ but $i\notin J(C')$ and $i\in J(C'')$. 
Then, there does not exist a CIF other than the \textit{Inclusive} CIF satisfying the axioms.
\end{proof}
On the other hand, we have that we need all the axioms to characterize the $\mathit{In}$ CIF:
\begin{propi}
The Inclusive CIF is not the only CIF satisfying three out of the four axioms (MON), (C), (I) and (EL)(i).
\end{propi}
\begin{proof}
The proof consists of 4 examples, each of which presents a CIF satisfying exactly three of the four axioms and none of these CIFs is the {\it Inclusive} one.
\begin{itemize}
\item  Consider the CIF $J(C)=\{i:0<|\{k|i\in C_{k}\}|<\frac{|N|}{2}\ or \ |\{k:i\in C_{k}\}|=N\}$. 
That is, individuals are deemed to be $J$ either if they are considered to be $J$ by less than half of all the individuals or if everybody agrees on that they are $J$.
This CIF does not satisfy (MON) because if an agent $i$ that is already a $J$ obtains more votes, eventually can exceed $\frac{|N|}{2}$ votes without requiring the approval of all the agents in $N$.
\item The CIF $J(C)=N$ for all $C\in\textbf{C}$ satisfies all the axioms except (C). 
Suppose that it satisfies (C) and $1\notin C_{i}$ for all $i\in N$. 
Then by (C), $1\notin J(C)=N$, a contradiction.
\item Consider the following CIF defined inductively. 
Let $J(C,0)$ be the set of all individuals for which there is a consensus in $C$ on that they are $J$.
Expand the set inductively by adding, at the t-th stage, those members for whom there is a consensus among $J(C,t-1)$ that they are $J$. 
In the case that $J(C,0)=\emptyset$ we define $J(C)=\{i:|\{k|i\in C_{k}\}|\geq 1\}$. 
This CIF verifies all the axioms except (I).
Consider for example the profiles $C=(\{1\},\{2\},\{3\})$ and $C'=(\{1,2\},\{2\},\{2,3\})$. 
Thus $J(C)=\{1,2,3\}$, but $J(C')=J(C',0)=J(C',1)=\{2\}$. 
Then (I) is not satisfied for agent $1$.
\item The Strong Liberal CIF satisfies all the axioms except (EL)(i).
For example, consider the profile $C=(\{2\},\{1\})$. 
Then $\mathit{StL}(C)=\emptyset$ violating (EL)(i).
\end{itemize} 
\end{proof}
We have a similar uniqueness result when we use part (ii) of (EL).
\begin{teo}
The Unanimity CIF $\mathit{U}$ is the only CIF that verifies (MON), (C), (I) and (EL)(ii).\footnote{Again, see Section 3 for the definition of $\mathit{U}$.}
\end{teo}
\begin{proof}
Clearly $\mathit{U}$ satisfies these four axioms. 
Consider a different CIF $J$ that also verifies them. 
Suppose there is a profile $C$ such that $i\notin C_{k} $ for some $k$ but $i\in J(C)$. 
By applying (MON) several times, we can find a profile $C'$ identical to $C$ with the exception that for every $j\neq i$, $i\in C^{'}_{j}$ so that $i\in J(C')$. 
Let $C^{''}$ be the profile such that for every $j\in N-\{k\}$, $C^{''}_{j}=N$ and $C^{''}_{k}=N-\{i\}$. 
By (C), $N-\{i\}\subseteq J(C'')$. So $J(C'')\in\{N-\{i\}, N\}$. 
Because this CIF satisfies part (ii) of (EL), it is not possible that $J(C'')=N$, so we have $J(C'')=N-\{i\}$. 
Again, we obtain a contradiction with (I), because agent $i$ is treated in the same way in profiles $C'$ and $C''$ but $i\in J(C')$ and $i\notin J(C'')$. 
Then, there does not exist a CIF other than the \textit{Unanimity} CIF satisfying the axioms.
\end{proof}
We have, as with the $\mathit{In}$ CIF, that axioms are independent in the characterization of the $\mathit{U}$ CIF:
\begin{propi}
The Unanimity CIF is not the only CIF that verifies three out of the four axioms (MON), (C), (I) and (EL)(ii).
\end{propi}
\begin{proof}
The proof consists again of 4 examples, each of which satisfies exactly three of the four axioms and none is the {\it Unanimity} CIF.
\begin{itemize}
\item Consider the CIF that first defines as $J$ the agents for whom there is a consensus in the profile $C$ that they are $J$. If this group is empty and $C_{1}\neq N$, then $J(C)=\{i:i\in C_{1}\ and \ i\notin C_{k}\ for\ all\ k\neq 1\}$. If $C_{1}=N$ then $J(C)=\emptyset$. This CIF verifies all the axioms except (MON). Consider the profile $C=(\{1\},\{2\},\{3\})$. Then $J(C)=\{1\}$. Now suppose that agent $1$ receives one vote more, as in $C'=(\{1\},\{1,2\},\{3\})$. We have $J(C')=\emptyset$. So (MON) is not verified.
\item Consider the following CIF defined inductively. 
Let $J(C,0)$ be the set of all individuals for which there is a consensus that they are $J$ in $C$. 
Expand the set inductively by adding, at the t-th stage, those members for whom there is a consensus among $J(C,t-1)$ that they are $J$. 
In the case that $J(C,t)=N$ for some $t>0$, we define $J(C)=\{i:|\{k|i\in C_{k}\}|=N\}$. 
This CIF verifies all the axioms except (I). 
Consider for example the profiles $C=(\{1,2\},\{1,2\},\{1,3\})$ and $C'=(\{1,2\},\{1,2\},\{3\})$. $J(C)=J(C',1)=\{1,2\}$ while $J(C')=\emptyset$. 
Then (I) is not satisfied for agent $2$.
\item The examples for (C) and (EL)(ii) are analogous to the ones used in the proof of Proposition 3.
\end{itemize} 
\end{proof}
From the uniqueness of the $\mathit{In}$ and $\mathit{U}$ CIFs we obtain the following result:
\begin{coro}
If a CIF verifies (MON), (C), (I) and (EL)(i) or (MON), (C), (I) and (EL)(ii); then it verifies (SYM).
\end{coro}

\section{Conclusions}
In this work we deal with the group identification problem and consider alternatives axioms to replace the liberal axiom used by K-R.
When we replace their liberal axiom by other axioms, some of which could be considered \textquotedblleft illiberal\textquotedblright, while another one could be seen as an extreme version of liberalism, we obtain some unique characterizations as well as an impossibility result.
There are many settings in which these alternative axioms can be meaningfully applied, particularly when we consider the spheres of decision allowed in liberal societies, which usually involve other individuals (children, students, etc.)\\
Our first result (Theorem 1) can be interpreted as indicating that when we allow at least two of the individuals to be decisive about their own identity or that of another specific agent, a unique rule satisfies all the proposed axioms. 
Furthermore, this rule allows them to be decisive only over themselves. 
The only CIF that satisfies this is the Strong Liberal CIF, the aggregator that K-R characterize in their work. 
The interesting thing is that this condition implies many other notions generally required to be satisfied by a ``fair'' aggregator. 
When we want a CIF to satisfy Decisiveness, we also obtain that this CIF verifies Monotonicity, Consensus and Independence.  
A social planner then knows that if she wants to implement this notion, it will come associated to those other properties.
\mbox{Proposition 2} shows that if we only allow the agents to be semidecisive over a specific agent (which could be himself), we find more rules satisfying all the axioms. 
That is, when we weaken the Decisiveness axiom, other rules can be chosen.
Finally, when the opinion of any agent is enough to ensure that the class of $J$ individuals is not the empty set nor the entire society, as it is asked in Extreme Liberalism, no rule will satisfy all the axioms.
The only possibility is to have two different CIFs, satisfying each of them only one part of this concept.\\
Another interesting aspect, is that the Strong Liberal CIF proposed by K-R, is the only one that verifies opposite axioms such as Liberalism and Decisiveness. This gives an idea of the power that this CIF has.
As in K-R, the idea here is not to find  ``the appropriate method'' for defining the members of a group. On the contrary, the goal is just to examine as thoroughly possible all the logical aspects of the problem. The axioms presented here should thus not be understood in normative terms, but just as an exploration of alternative notions, somewhat close to some intuitions on what some societies allow.

\end{document}